\documentclass[11pt]{article}

\usepackage{geometry}
\usepackage{amssymb}
\usepackage{mathrsfs}
\usepackage{amsmath}
\usepackage[latin1]{inputenc}
\newtheorem{theorem}{Theorem}[section]
\newtheorem{corollary}[theorem]{Corollary}

\newtheorem{claim}[theorem]{Claim}
\newtheorem{definition}[theorem]{Definition}

\def\squarebox#1{\hbox to #1{\hfill\vbox to #1{\vfill}}}
\newcommand{\qed}{\hspace*{\fill}
\vbox{\hrule\hbox{\vrule\squarebox{.667em}\vrule}\hrule}\smallskip}

\newenvironment{proof}{\noindent{\bf Proof:~~}}{\(\qed\)}

\newcommand{\ROG}{ROG}

\begin{document}

\title{On the Greedy Algorithm for Combinatorial Auctions\\ with a Random Order}

\author {Shahar Dobzinski \and Ami Mor}
\maketitle

\begin{abstract}
In this note we study the greedy algorithm for combinatorial auctions with submodular bidders. It is well known that this algorithm provides an approximation ratio of $2$ for every order of the items. We show that if the valuations are vertex cover functions and the order is random then the expected approximation ratio imrpoves to $\frac 7 4$.
\end{abstract}

\section{Introduction}

In a combinatorial auction there is a set $M$ items ($|M|=m$) for sale. There is also a set of $N$ bidders ($|N|=n$). Each bidder $i$ has a valuation function $v_i:2^M\rightarrow \mathbb R$ that denotes the value of bidder $i$ for every possible subset of the items. We assume that the valuation functions are monotone (for all $S\subseteq T$, $v_i(S)\geq v_i(T)$) and normalized ($v_i(\emptyset)=0$). For an item $j$ and a bundle $S$, we use the notation $v(j|S)$ to denote the marginal value of item $j$ given a bundle $S$: $v(\{j\}+S)-v(S)$. The goal is to find an allocation of the items that maximizes the welfare $\Sigma_iv_i(S_i)$. We assume that access to the valuation function is done via \emph{value queries}: for a bundle $S$, what is $v(S)$? Our goal is to (approximately) maximize the welfare using only $poly(m,n)$ value queries.

In this note we consider the case where the valuations of the bidders are submodular. A valuation $v$ is \emph{submodular} if $v(j|S)\geq v(j|T)$ for every $S\subseteq T$. The problem of combinatorial auctions with submodular valuations was first introduced by Lehmann, Lehmann, and Nisan \cite{LLN01} who showed that the following simple greedy algorithm provides an approximation ratio of $2$: go over all items one by one at some arbitrary order, allocate each item to the bidder that maximizes the expected marginal value of the item given the set of items that he won so far.

Follwing \cite{LLN01}, the problem was extensively studied in the literature and is now well understood: there exists an $\frac e {e-1}$ approximation algorithm that makes polynomiallly many queries \cite{V08} and this is the best possible \cite{MSV08, F98}.

Here we return to the classic greedy algorithm of $\cite{LLN01}$ and ask whether the expected approximation ratio improves when we use a random ordering of the items instead of an arbitrary one. While we do not know if this is the case in general, we are able to answer this question for a special case:

\vspace{0.1in} \noindent \textbf{Theorem: } If all valuations are vertex cover valuations then the random greedy algorithm provides an approximation ratio of $\frac 7 4$.

\vspace{0.1in} \noindent Vertex cover functions are defined as follows: let $G=(V,E)$ be a graph with $|V|=m$. We associate the set of vertices $V$ with the set of items $M$ in the combinatorial auction. We let $v(S)$ be the number of edges that have at least one end point in $S$. It is easy to see that every vertex cover function is also submodular (and in fact is also a coverage valuation)\footnote{The name \emph{monotone cut functions} is also sometimes associated with these functions for the following reason: consider two bidders with vertex cover valuations that are defined on the same graph $G$. Consider any allocation of all items $(S,M-S)$. The value of this allocation is the number of vertices plus twice the number of edges in the cut defined by $S$.}.

We also show that the random greedy algorithm does not provide an approximation ratio better than $\frac {240} {177}\approx 1.356$ on vertex cover functions. Note that the problem is APX-hard \cite{DNS05}.

There are several advantages to our approach. First, there is the obvious advantage of showing that simple and practical algorithms provide a good approximation ratio. Furthermore, even though there exists a polynomial time algorithm that obtains an approximation ratio of $\frac e {e-1}$ \cite{V08}, our work contibutes to the growing body of literature attempting to obtain faster algorithms for problems in submodular optimization with comparable approximation guarantees (see, e.g., \cite{BV14, CJV15, BFS15}).

In addition, the greedy algorithm is in a sense an online algorithm: items arrive one by one and we have to allocate the items immediately. In this model it is known that the greedy algorithm is essentially optimal \cite{KPV13} (although an algorithm with a slightly better approximation ratio of $ 2 -\frac 1 n$ exists \cite{DS06}). Hence it is natural to explore hybrid models, for example, in the spirit of the famous secretary problem, when the set of items is determined adversarially but their order is random.

A final advantage of the random greedy algorithm that we would like to discuss here is an application to auction design. Christodoulou et al \cite{CKS08} study simultaneous second price auctions: each bidder submits a bid for every item, he wins the set of items for which his bid was maximal, and pays the sum of the second-highest bids of every item he won. They show that subject to a certain no-overbidding condition every Nash equilibrium provides a $2$ approximation to the optimal welfare. Furthermore, the greedy algorithm can be easily adapted to produce a Nash equilibrium with that approximation ratio. It is not known how to compute a Nash equilibrium with an approximation ratio better than $2$ with only polynomially many value queries (although this is possible for several special cases, see \cite{DFK15}), and our work suggests the analysis of the ranodm greedy algorithm as a possible attack in this direction.

We leave open the question of determining the approximation ratio of the random greedy algorithm for combinatorial auctions with general submodular bidders. A very interesting open question is to understand whether the random greedy algorithm can obtain an approximation ratio better than $\frac e {e-1}$ for interesting special cases, such as combinatorial auctions with budget additive valuations. The current best known algorithm guarantees an approximation ratio of $\frac 4 3$ via a complicated iterative rounding scheme \cite{CG08}. Can a mathcing or even better ratio be obtained by the random greedy algorithm?

\section{Preliminaries}

We consider the following Random Greedy algorithm:

\subsubsection*{$\ROG\left(v_1...v_n,M\right)$}
\begin{enumerate}
\item Choose uniformly at random a permutation $\sigma$ on the order of the items.
\item Let $S_1^0=...=S_n^0 = \emptyset$.
\item For $t=1..m$:
\begin{enumerate}
\item Denote by $j=\sigma(t)$ the $t\text{'th}$ item in the permutation.
\item Let $i$ be a player whose marginal value of $j$ is maximal. I.e., $i\in \arg\max_iv_i\left(j|S_i\right)$.
\item For every $i'\neq i$ $S_{i'}^t=S_{i'}^{t-1}$. Set $S_i^t = S_i^{t-1} \cup \left\{j\right\}$.
\end{enumerate}
\item Return $S_1^m...S_n^m$.
\end{enumerate}

We analyze the approximation ratio of the algorithm in the case where all functions are vertex cover functions:

\begin{definition}[vertex cover functions] 
We say that $v$ is a \emph{vertex cover function} if there exists an undirected graph $G=(V,E)$ such that
for every $S\subseteq V$ we have that $v\left(S\right)=|\left\{e=\left(i,j\right)|i\in S \text{ or } j\in S\right\}|$.
\end{definition}

From now on we assume that all the $v_i$'s are vertex cover functions. I.e., there are graphs $G_1...G_n$ such that for every $i$, $v_i$ is the vertex cover function of $G_i$. 

\section{The Approximation Ratio of the Random Greedy Algorithm}
\subsection{A Lower Bound}
We start the analyis of the approximation ratio of the random greedy algorithm with a lower bound.
\begin{claim} \label{COUNTER}
There is an instance for which $\frac{E[\ROG\left(v_1...v_n,M\right)]}{OPT} \leq \frac{177}{240}+o\left(1\right)$.
\end{claim}
\begin{proof}
We produce an instance with three players and $m$ is odd. $G_i=\left(V,E_i\right)$. 
\begin{itemize}
\item $G_1$ is a star graph: $E_1 = \left\{\left(1,m\right),\left(2,m\right)...\left(m-1,m\right)\right\}$.
\item $G_2$ is a line graph: $E_2 = \left\{\left(1,2\right),\left(3,4\right)...\left(m-2,m-1\right)\right\}$.
\item $G_2$ is also a line graph: $E_3 = \left\{\left(2,3\right),\left(4,5\right)...\left(m-3,m-2\right)\right\}$.

\end{itemize}
In case of ties, we prefer to give items to player $1$, rather than to player $2$, rather than to player $3$.

The optimal allocation is $OPT_1 = \left\{m\right\}$, $OPT_2 = \left\{1,3,...m-2\right\}$ and $OPT_3 = \left\{2,4,...m-1\right\}$.
$\left(OPT_1,OPT_2,OPT3\right)$ uses all the edges, so it has to be the optimal solution.
The value of the optimal allocation is $2m-1$.

\begin{claim} \label{player1}
$E[v_1\left(S_1^m\right)]=m-1$.
\end{claim}
\begin{proof}
Player $1$ always takes item $m$, and then his value is $v_1\left(S_i^m\right)\geq v_1\left(m\right) = m-1$.
On the other hand $v_1\left(M\right) = m-1$, so $v_1\left(S_1^m\right) = m-1$ always.
\end{proof}

\begin{claim} \label{player2}
$E[v_2\left(S_2^m\right)]=\frac{m-1}{3}$ .
\end{claim}
\begin{proof}
For $i=1...\frac{m-1}{2}$, let $A_i$ be an indicator random variable, that indicates whether we assigned items $2i-1$, or $2i$ to player $2$.
Note that $v_2\left(S_2^m\right) = \sum_{i=1}^\frac{m-1}{2} A_i$.
If $m$ appeared after both $2i-1$ and $2i$, the algorithm assigned $2i-1$, $2i$ and $m$ to player $1$. Otherwise, the algorithm assigned either
$2i-1$ or $2i$ to player $2$.
 Therefore, $Pr[A_i=1] = \frac{2}{3}$, and we have 
$$E[v_2\left(S_2^m\right)] = \sum_{i=1}^\frac{m-1}{2} E[A_i]=\frac{m-1}{3}$$
\end{proof}

\begin{claim} \label{player3}
$E[v_3\left(S_3^m\right)]=\frac{17\left(m-3\right)}{120}$.
\end{claim}
\begin{proof}
For $i=1...\frac{m-3}{2}$ let $A_i$ be an indicator random variable that indicates whether we assigned item $2i$ or $2i+1$ to player $3$.
I.e $A_i$ indicates if we used the edge $\left(2i,2i+1\right)$ in the graph $G_3$.
Note that $v_3\left(S_3^m\right) = \sum_{i=1}^\frac{m-1}{2} A_i$.

Consider the item $2i$ and denote $t = \sigma^{-1}\left(j\right)$. $v_1\left(2i\right)=v_2\left(2i\right)=v_3\left(2i\right)=1$.
The algorithm assigns item $2i$ to player only if $v_1\left(2i|S_i^t\right)=v_2\left(2i|S_i^t\right)=0$ and $v_3\left(2i|S_i^t\right)=1$.
If $v_2\left(2i|S_i^t\right)=0$ then it means that
\begin{enumerate}
\item $2i-1$ appears before $2i$ in the permutation $\sigma$, because the algorithm used edge $\left(2i-1,2i\right)$.
\item $m$ appears before $2i-1$, because otherwise we would have given item $2i-1$ to player $1$.
\end{enumerate}
Therefore, the order of items $2i-1$, $2i$ and $m$ must be $\left(m,2i-1,2i\right)$. 
Similarly, the order of items $2i+1$, $2i+2$ and $m$ must be $\left(m,2i+2,2i+1\right)$. 

Let $B_i$ be the event that the order of items $2i-1,2i,m$ is $\left(m,2i-1,2i\right)$.
Let $C_i$ be the event that the order of items $2i+1,2i+2,m$ is $\left(m,2i+2,2i+1\right)$.
$$Pr[A_i]=Pr[B_i=1 \text{ or } C_i=1]= Pr[B_i] + Pr[C_i] - Pr[ B_i=1 \text{ and } C_i=1 ]$$

Note that $B_i=1$ and $C_i=1$ if and only if all $3$ conditions are true:
\begin{enumerate}
\item $m$ is the first item among $\left(m,2i-1,2i,2i+1,2i+2\right)$. This happens with probability $\frac{1}{5}$.
\item $2i-1$ appears before item $2i$. This happens with probability $\frac{1}{2}$.
\item $2i+2$ appears before item $2i+1$. This happens with probability $\frac{1}{2}$.
\end{enumerate}
Note that these conditions are independent so $Pr[ B_i=1 \text{ and } C_i=1 ] = \frac{1}{20}$.

Overall, we have $Pr[A_i] = Pr[B_i] + Pr[C_i] - Pr[ B_i=1 \text{ and } C_i=1 ] = \frac{1}{6} + \frac{1}{6} - \frac{1}{20}=\frac{17}{60}$.
Since there are $\frac{m-3}{2}$ different $A_i$, then $E[v_3\left(S_3^m\right)]=\frac{17\left(m-3\right)}{120}$.
\end{proof}

Back to the proof of claim \ref{COUNTER}:
$$E[\left(\ROG\right)] = E[v_1\left(S_1^m\right)]+E[v_2\left(S_2^m\right)]+E[v_3\left(S_3^m\right)] = \left(m-1\right)\left(1+\frac{1}{3}+\frac{17}{120} + O\left(1\right)\right)$$
Therefore,
$\frac {E[\left(\ROG\right)]}{OPT} = \frac{177}{240} + o\left(1\right)$.
\end{proof}

\subsection{The Upper Bound: A Warm Up}

As a warm up, we first show that when all valuations are vertex cover functions then the random greedy algorithm provides an approximation ratio of $2$. Of course, in \cite{LLN01} it was shown that \emph{every} ordering provides this approximation ratio, but the proof we give provides an illustration of the properties of vertex cover functions that we use in our main result.

\begin{theorem} \label{0.5}
Let $v_1...v_n$ be $n$ vertex cover functions. $E[\ROG\left(v_1...v_n,M\right)]\geq \frac{1}{2}\sum v_i\left(OPT_i\right)$.
\end{theorem}
\begin{proof}
Consider a certain item $j$, for which $i=O\left(j\right)$. Vertex $j$ has $v_{i}\left(j\right)$ neighbours in the grapsh $G_i$.
Let $B\left(j\right)$ be a random variable representing the number of $j$'s neighbours that appeared in the random order before $j$.
$B\left(j\right)$ distributes uniformly between $[0..v_{i}\left(j\right)]$.
Therefore, 
\begin{equation} \label{EXP1}
E[B\left(j\right)] = \frac{1}{2}v_{i}\left(j\right)
\end{equation}
Note that in vertex cover functions, the marginal value of a vertex is the total 
number of neighbours it has, that has not yet taken by the algorithm.
Therefore (we use the standard notation $\delta_G\left(v\right)$ for the set of neighbours of $v$ in a graph $G$),
\begin{equation} \label{eq1}
v_i\left(j|S_i\right) = v_i\left(j\right) - |S_i\cap\delta_{G_i}\left(j\right)|
\end{equation}
Since $S_{i} \cap \delta_{G_{i}}\left(j\right)$ contains only vertices that appear before $j$,
we have that $|S_{i}\cap\delta_{G_{i}}\left(j\right)|\leq B\left(j\right)$.
Therefore, by \eqref{EXP1} and \eqref{eq1} : $E[v_i\left(j|S_i\right)] \geq \frac{1}{2} v_i\left(j\right)$. 
The algorithm chooses at every step the player that maximizes $v_{i'}\left(j|S_{i'}\right)$. 
Therefore, $E[\max_{i'\in N}v_i'\left(j|S_{i'}\right)] \geq E[v_i\left(j|S_i\right)] \geq \frac{1}{2} v_i\left(j\right)$.
By summing over all items we get that the value of the allocation that the algorithm returns is at least
$\frac{1}{2}\Sigma_j v_{O\left(j\right)}\left(j\right)$, which is at least half of the value of the optimal solution.
\end{proof}

In the next subsection we will show that the random greedy algorithm indeed provides an approximation ratio better than $2$. The improved analysis takes advantage on some slackness in the proof above.
For example, if the approximation ratio is no better than $2$, then the proof above implies that for every vertex $j$, $|S_i\cap\delta_{G_i}\left(j\right)|=B\left(j\right)$, i.e., the algorithm used all the edges of 
vertex $j$ in $G_{O\left(j\right)}$ that appeared before it in the random order. However, if that is the case, then we should count the contribution of these edges in the value of the solution that the algorithm outputs. In a similar manner, if the approximation ratio is no better than $2$, the analysis implies that the value that the algorithm gained from item $j$ is close to $v_O\left(j\right)\left(j|S_{O\left(j\right)}\right)$.
Therefore, for every item $j$, there is another player, $i'$ for whom the marginal value of item $j$ is close to the marginal value of player $O\left(j\right)$. Again, we can use the other player to claim that the expected value of the solution is higher.

\subsection{The Main Result}

For the analysis we fix an optimal allocation $(OPT_1,...,OPT_n)$ and let $OPT=\sum v_i\left(OPT_i\right)$.
\begin{theorem} \label{greedymain}
Let $v_1...v_n$ be $n$ vertex cover functions. $E[\ROG\left(v_1...v_n,M\right)]\geq \frac{4}{7}\sum v_i\left(OPT_i\right)$.
\end{theorem}
\begin{proof}
In the proof we use the notation $O\left(j\right)$ to denote the player that gets item $j$ in the optimal allocation (i.e., $j\in OPT_i$) and  $A\left(j\right)$ to denote the player who gets item $j$ in the algorithm. Notice that $A\left(j\right)$ is a random variable.

Denote by $OPT^t$ the marginal value of the optimal allocation, given that the algorithm already allocated items $\sigma\left(1\right)...\sigma\left(t-1\right)$.
$OPT^t = \sum v_i\left(OPT_i \cap \{\sigma(t),\ldots ...\sigma(m)\}|S_i^t\right)$. Notice that $OPT^0 = OPT$ and $OPT^m = 0$.

We say that an edge $e=\left(j_1,j_2\right)\in G_i$ has been taken by the algorithm at step $t$, if $j_1=\sigma\left(t\right)$, $A\left(j_1\right)=i$, 
and $j_2\notin S_i^{t-1}$. I.e., the algorithm gave player $i$ one vertex adjacent to the edge, while the other vertex has not yet been given to player $i$.

For every item $j$, let $C\left(j\right)=\arg\max_{i\neq O\left(j\right)}v_i\left(j\right)$.
I.e., $C\left(j\right)$ is the largest competitor to take the item from $O\left(j\right)$. In the following we denote by $b_C\left(j\right)$ the number of edges adjacent to $j$ in the graph $G_{C\left(j\right)}$ that were used by the algorithm before item $j$. I.e., $b_C\left(j\right)=v_{C\left(j\right)}\left(j\right)-v_{C\left(j\right)}\left(j|S_{C\left(j\right)}^{t-1}\right)$.
Similarly, denote by $b_O\left(j\right)$  the number of edges, adjacent to $j$ in the graph $G_{O\left(j\right)}$,
used by the algorithm before item $j$. We use a counting argument to bound the quality of the solution that the algorithm outputs:
\begin{claim} \label{BEFORE}
$\ROG\left(v_1...v_n,M\right)\geq \sum_{j=1}^M \left(b_C\left(j\right)+b_O\left(j\right)\right)$.
\end{claim}
Proof for this claim will be provided later. Note that  $\sum b_O\left(j\right) = \sum v_{O\left(j\right)}\left(j\right)-\sum v_{O\left(j\right)}\left(j|S_{O\left(j\right)}^{t-1}\right)$.
By submodularity, $\sum v_{O\left(j\right)}\left(j\right) \geq OPT$. The algorithm at each step chooses the item that maximizes the marginal utility. Denote by $\ROG^t$ the value of all the players at iteration $t$: $\ROG^t = \sum v_i\left(S_i^t\right)$. Therefore, $\ROG^{t}-\ROG^{t-1} \geq v_{O\left(j\right)}\left(j|S_{O\left(j\right)}^{t-1}\right)$. Therefore,
\begin{equation*}
\sum b_O\left(j\right) \geq OPT - \ROG\left(v_1...v_n,M\right)
\end{equation*}
Combined with claim \ref{BEFORE}, we get that:
\begin{corollary} \label{COR}
$b_C\left(j\right)\leq 2\ROG\left(v_1...v_n,M\right) - OPT$
\end{corollary}

Let $\ROG\left(j\right)$ be a random variable representing value the algorithm gains from item $j$. I.e., for $j=\sigma\left(t\right)$, then $\ROG\left(j\right) = \ROG^{t}-\ROG^{t-1}$. 
We present an argument in which a competitor for item $j$ with higher value increases the expectancy of the value of the algorithm from item $j$.

\begin{claim} \label{POS}
\begin{itemize}
\item If $v_{O\left(j\right)}\left(j\right) \geq v_{C\left(j\right)}\left(j\right)$, then
$$E\left(\ROG\left(j\right)\right)\geq \left(\frac{v_{O\left(j\right)}\left(j\right)}{2} + \frac{v_{C\left(j\right)}\left(j\right)^2+v_{C\left(j\right)}\left(j\right)}{2\left(v_{O\left(j\right)}+1\right)}-b_C\left(j\right)\right)$$
\item If $v_{C\left(j\right)}\left(j\right) \geq v_{O\left(j\right)}\left(j\right)$, then
$$E\left(\ROG\left(j\right)\right)\geq v_{C\left(j\right)}\left(j\right) -b_C\left(j\right)$$
\end{itemize}
\end{claim}
The proof for this claim will also be provided later.

In contrast to claim \ref{POS} we claim that if $v_{C\left(j\right)}\left(j\right)$ is small, then roughly speaking, the algorithm
almost always gives item $j$ to $O\left(j\right)$. That is because if item $j$ appears in the beginning of the random order, then the marginal 
value of $j$ to player $O\left(j\right)$ is bigger than $v_{C\left(j\right)}\left(j\right)$.

for an item $j=\sigma\left(t\right)$, let $LOSS\left(j\right)$ be the decrease in the value of the optimal, given that we
gave item $j$ to $A\left(j\right)$: $LOSS\left(j\right) = OPT^{t-1} - OPT^{t}$.
\begin{claim} \label{NEG}
If $v_{O\left(j\right)}\left(j\right) \geq v_{C\left(j\right)}\left(j\right)$, then
$$E[\ROG\left(j\right)] \geq  E[LOSS\left(j\right)] - \frac{v_{C\left(j\right)}\left(j\right)^2+v_{C\left(j\right)}\left(j\right)}{v_{O\left(j\right)\left(j\right)}+1}$$
If $v_{C\left(j\right)}\left(j\right) \geq v_{O\left(j\right)}\left(j\right)$, then
$$E[\ROG\left(j\right)] \leq  E[LOSS\left(j\right)] - v_{O\left(j\right)}\left(j\right)$$.

\end{claim}

The proof for this claim will also be provided later. First, we show how to derive theorem \ref{greedymain} using claims \ref{BEFORE}, \ref{POS} and \ref{NEG}.

We split our analysis to two sets of items. $M_1$ is the set of items for which $v_{O\left(j\right)}\left(j\right) \geq v_{C\left(j\right)}\left(j\right)$.
$M_2$ is the rest of the items.
By claim \ref{POS} on items in $M_1$ we have,
\begin{equation*}
2E[\sum_{j\in M_1}\ROG\left(j\right)]\geq 2\sum_{j\in M_1}\left(\frac{v_{O\left(j\right)}\left(j\right)}{2} + \frac{v_{C\left(j\right)}\left(j\right)^2+v_{C\left(j\right)}\left(j\right)}{2\left(v_{O\left(j\right)}\left(j\right)+1\right)}-b_C\left(j\right)\right)
\end{equation*}
By claim \ref{NEG} we have,
\begin{equation*}
E[\sum_{j\in M_1}\ROG\left(j\right)]\geq E[\sum_{j\in M_1}LOSS\left(j\right)] - \frac{v_{C\left(j\right)}\left(j\right)^2+v_{C\left(j\right)}\left(j\right)}{\left(v_{O\left(j\right)}+1\right)}
\end{equation*}

Summing these equation and rearranging:
\begin{equation} \label{M1}
E[\sum_{j\in M_1}\ROG\left(j\right)] \geq \frac{1}{3}E[\sum_{j\in M_1}\left(v_{O\left(j\right)}\left(j\right)+LOSS\left(j\right)-2b_C\left(j\right)\right)]
\end{equation}

Our next goal is to give a similar bound for items in $M_2$. We apply claim \ref{POS} on all items in $M_2$. Note that this time we use the second part of the claim:
\begin{equation} \label{5}
\frac{2}{3}E[\sum_{j\in M_2}\ROG\left(j\right)]\geq \frac{2}{3}E\left(\sum_{j\in M_2}v_{C\left(j\right)}\left(j\right)-b_C\left(j\right)\right) 
\end{equation}
Using claim \ref{NEG} we have:
\begin{equation} \label{6}
\frac{1}{3}E[\sum_{j\in M_2}\ROG\left(j\right)]\geq \frac{1}{3}E\left(\sum_{j\in M_2}E[LOSS\left(j\right)]- v_{O\left(j\right)}\left(j\right)\right) 
\end{equation}
Summing \eqref{5} and \eqref{6}, we have
\begin{equation*}
E[\sum_{j\in M_2}\ROG\left(j\right)]\geq \frac{1}{3}E\left(\sum_{j\in M_2}E[LOSS\left(j\right)]- v_{O\left(j\right)}\left(j\right) + 2v_{C\left(j\right)}\left(j\right) -2b_C\left(j\right)\right) 
\end{equation*}
Since we are looking at $j\in M_2$, $v_{C\left(j\right)}\left(j\right)\geq v_{O\left(j\right)}\left(j\right)$. Therefore,
\begin{equation} \label{M2}
E[\sum_{j\in M_2}\ROG\left(j\right)] \geq \frac{1}{3}E[\sum_{j\in M_2}\left(v_{O\left(j\right)}\left(j\right)+LOSS\left(j\right)-2b_C\left(j\right)\right)]
\end{equation}
Combining \eqref{M1} and \eqref{M2}:
\begin{equation*}
E[\sum_{j\in M}\ROG\left(j\right)]\geq \frac{1}{3}E\left(\sum_{j\in M} v_{O\left(j\right)}\left(j\right) + E[LOSS\left(j\right)] - 2b_C\left(j\right)\right) 
\end{equation*}
Since $\sum v_{O\left(j\right)}\geq OPT$, $\sum\left(LOSS\left(j\right)\right)=OPT$, and by corollary \ref{COR},
We have  
$$E[\sum_{j\in M}\ROG\left(j\right)]\geq \frac{1}{3}\left(OPT+OPT-4\ROG + 2OPT\right) $$
It follows that $E[\sum_{j\in M}\ROG\left(j\right)]\geq \frac{4}{7} OPT$.
\end{proof}

\subsubsection*{Proof of claim \ref{BEFORE}}
We will show that $\ROG\left(v_1...v_n,M\right)\geq \sum_{j=1}^M \left(b_C\left(j\right)+b_O\left(j\right)\right)$,
by proving that $\sum_{j=1}^M \left(b_C\left(j\right)+b_O\left(j\right)\right)$ counts partial set of all the edges taken by the algorithm.

We denote each edge taken by the algorithm by $\left(j_1,j_2,i\right)$, where $\left(j_1,j_2\right)\in G_i$.
Consider an item $j$ such that $j=\sigma\left(t\right)$.
Let $B_O\left(j\right)$ the set of edges adjacent to $j$ in the graph $G_{C\left(j\right)}$,
used by the algorithm before item $j$. 
I.e., $B_O\left(j\right)=\left\{\left(j_1,j,i\right)|i=O\left(j\right) \wedge j_1\in S_i^{t'} \wedge \left(t'<t\right) \right\}$, 
and $b_O\left(j\right)=|B_O\left(j\right)|$.
Similarly, Let $B_C\left(j\right)=\left\{\left(j_1,j,i\right)|i=C\left(j\right) \wedge  j_1\in S_i^{t'} \wedge \left(t'<t\right) \right\}$
and $b_C\left(j\right)=|B_c\left(j\right)|$.

Notice that for every $j$, $B_C\left(j\right)$ and $B_O\left(j\right)$, have only edges taken by the algorithm.
Moreover, all the sets $B_C\left(1\right)...B_C\left(m\right),B_O\left(1\right)...B_O\left(m\right)$ are disjoint.
Therefore, $\sum_{j=1}^M \left(b_C\left(j\right)+b_O\left(j\right)\right)$ counts only edges used by the algorithm, and counts every edge
at most once. The claim follows.

\subsubsection*{Proof of claim \ref{POS}}
We start with the following claim:
\begin{claim} \label{technical}
Let $X$ be a uniformly distributed random variable between $[0...x]$ and let $y$ be a constant. 
If $x\geq y$  then $E[\max \left(X,y\right)]\geq \frac{x}{2} + \frac{y^2+y}{2\left(x+1\right)}$.
If $y\geq x$ then $E[\max \left(X,y\right)]= y$
\end{claim}
\begin{proof}
Computing directly the expectation, we get
\begin {eqnarray*}
E[\max \left(X,y\right)]& = &\sum_{i=0}^y\frac{y}{x+1} + \sum_{i=y+1}^x \frac{i}{x+1}\\
& = & \frac{\left(y+1\right)y}{x+1} + \frac{\left(x-y\right)\left(x+y+1\right)}{2\left(x+1\right)}\\
& = & \frac{x}{2} + \frac{y^2+y}{2\left(x+1\right)}
\end{eqnarray*}

On the other hand, if $y\geq x$, then $\max \left(X,y\right) = y$ always.
\end{proof}

To prove claim \ref{POS} we need to show:
\begin{itemize}
\item If $v_{O\left(j\right)}\left(j\right) \geq v_{C\left(j\right)}\left(j\right)$, then
$$E\left(\ROG\left(j\right)\right)\geq \left(\frac{v_{O\left(j\right)}\left(j\right)}{2} + \frac{v_{C\left(j\right)}\left(j\right)^2+v_{C\left(j\right)}\left(j\right)}{2\left(v_{O\left(j\right)}+1\right)}-b_C\left(j\right)\right)$$
\item If $v_{C\left(j\right)}\left(j\right) \geq v_{O\left(j\right)}\left(j\right)$, then
$$E\left(\ROG\left(j\right)\right)\geq v_{C\left(j\right)}\left(j\right) -b_C\left(j\right)$$
\end{itemize}
Suppose that item $j=\sigma\left(t\right)$. The algorithm can choose to give the item to either $O\left(j\right)$ or to $C\left(j\right)$.
\begin{eqnarray*}
\ROG\left(j\right) &\geq & \max \left\{v_{O\left(j\right)}\left(j|S_{O\left(j\right)}^{t-1}\right),v_{C\left(j\right)}\left(j|S_{C\left(j\right)}^{t-1}\right)\right\} \\ 
& = &\max \left\{v_{O\left(j\right)}\left(j\right)-b_O\left(j\right),v_{C\left(j\right)}\left(j\right)-b_C\left(j\right)\right\} \\
& \geq & \max\left\{v_{O\left(j\right)}\left(j\right)-b_O\left(j\right),v_{C\left(j\right)}\left(j\right)\right\}-b_C\left(j\right) 
\end{eqnarray*}
We repeat an argument we used in claim \ref{0.5}: $b_O\left(j\right)$ can count only edges that appear before $j$ in the permutation $\sigma$.
 $v_{O\left(j\right)}\left(j\right) - b\left(j\right)\geq |\{\sigma(t+1),\sigma(m)\} \cap \delta_{G_{O\left(j\right)}}|$. 
 Set $X = |\{\sigma(t+1),\ldots, \sigma(m)\} \cap \delta_{G_{O\left(j\right)}}|$. 
We have: $$\ROG\left(j\right) \geq \max\left(X,v_{C\left(j\right)}\left(j\right)\right) - b_C\left(j\right)$$ 
Since $X$ distributes uniformly between $[0..v_{O\left(j\right)}\left(j\right)]$, we can apply claim \ref{technical} to bound 
$\max\left(X,v_{C\left(j\right)}\left(j\right)\right)$, and claim \ref{POS} follows.

\subsubsection*{Proof of claim \ref{NEG}}
We use a claim similar to the claim used for the proof that the greedy algorithm return a $\frac{1}{2}$ approximation of the 
maximal welfare in \cite{LLN01}.

\begin{claim} \label{classic}
Fix a permutation $\sigma$. Consider $j=\sigma\left(t\right)$.
If the algorithm assigned item $j$ to $O\left(j\right)$, then $LOSS\left(j\right) \leq \ROG\left(j\right)$.
If the algorithm assigned item $j$ to $i\neq O\left(j\right)$ then
$LOSS\left(j\right) \leq \ROG\left(j\right) + v_{O\left(j\right)}\left(j|S_{O\left(j\right)}^{t-1}\right)$.
\end{claim}
The proof for this claim appears in the appendix. Notice that this claim is true for any permutation $\sigma$.

We will now derive claim \ref{NEG}. Recall that we want to show that:
\begin{itemize}
\item If $v_{O\left(j\right)}\left(j\right) \geq v_{C\left(j\right)}\left(j\right)$, then
$$E[\ROG\left(j\right)] \geq  E[LOSS\left(j\right)] - \frac{v_{C\left(j\right)}\left(j\right)^2+v_{C\left(j\right)}\left(j\right)}{v_{O\left(j\right)}+1}$$
\item If $v_{C\left(j\right)}\left(j\right) \geq v_{O\left(j\right)}\left(j\right)$, then
$$E[\ROG\left(j\right)] \leq  E[LOSS\left(j\right)] - v_{O\left(j\right)}\left(j\right)$$
\end{itemize}

We start with the case where $v_{O\left(j\right)}\left(j\right) \geq v_{C\left(j\right)}\left(j\right)$.
If $v_{O\left(j\right)}\left(j|S^t_i\right)> v_{C\left(j\right)}\left(j\right)$, the algorithm gives item $j$ to $O\left(j\right)$.
$v_{O\left(j|S_i^t\right)} = v_{O\left(j\right)}-b_O\left(j\right)$, and as we have seen before, $b_O\left(j\right)$ can be bounded by a uniform variable.
\begin{eqnarray*}
Pr[v_{O\left(j\right)} \left(j|S_i^t\right) > v_{C\left(j\right)}\left(j\right)] & = & Pr[b_O\left(j\right) < \left(v_{O\left(j\right)}\left(j\right) - v_{C\left(j\right)}\left(j\right)\right)] \\
& \geq &\frac{v_{O\left(j\right)}\left(j\right) - v_{C\left(j\right)}\left(j\right)} {v_{O\left(j\right)}\left(j\right)+1}
\end{eqnarray*}
Therefore, the probability that the algorithm did not give item $j$ to $O\left(j\right)$ is at most $\frac{1 + v_{C\left(j\right)}\left(j\right)} {v_{O\left(j\right)}\left(j\right)+1}$,
and in that case $v_{O\left(j\right)}\left(j|S_{O\left(j\right)}^{t-1}\right) \leq  v_{C\left(j\right)}\left(j\right)$.
We use claim \ref{classic} to bound $E[LOSS\left(j\right)]$:
\begin{eqnarray*}
E[LOSS\left(j\right)] &= &E[LOSS\left(j\right) | j\in O\left(j\right)]\cdot Pr[j\in O\left(\right)] + E[LOSS\left(j\right) | j\notin O\left(j\right)]\cdot Pr[j\notin O\left(j\right)] \\ 
& \leq & E[ALG\left(j\right)\cdot Pr[j\in O\left(\right)] + \left(E[\ROG\left(j\right)]+ v_{C\left(j\right)}\left(j\right)\right)\cdot Pr[j\notin O\left(j\right)]\\ 
& = & E[ALG\left(j\right)] + \frac{v_{C\left(j\right)}\left(j\right) + v_{C\left(j\right)}\left(j\right)^2} {v_{O\left(j\right)}\left(j\right)+1} 
\end{eqnarray*}

Suppose that $v_{O\left(j\right)}\left(j\right) \leq v_{C\left(j\right)}\left(j\right)$.
By claim \ref{technical}, $LOSS\left(j\right) \leq \ROG\left(j\right) + v_{O\left(j\right)} \left(j|S_{O\left(j\right)}^{t-1}\right)$.
Since $v_{O\left(j\right)} \left(j|S_{O\left(j\right)}^{t-1}\right)\leq v_{O\left(j\right)} \left(j\right)$,
we have 
$E[\ROG\left(j\right)] \leq  E[LOSS\left(j\right)] - v_{O\left(j\right)}\left(j\right)$.

\bibliographystyle{plain}
\bibliography{bib}

\end{document}